\documentclass{amsart}
\usepackage{amssymb}
\usepackage{amsfonts}

\newtheorem{theorem}{Theorem}
\theoremstyle{plain}

\newtheorem{lemma}{Lemma}

\newtheorem{proposition}{Proposition}
\newtheorem{remark}{Remark}

\numberwithin{equation}{section}
\input{tcilatex}

\begin{document}
\title[Derivation]{The Grillakis-Machedon-Margetis Second order corrections
to mean field evolution for weakly interacting Bosons in the case of 3-body
interactions}
\author{Xuwen Chen}
\address{University of Maryland, College Park}
\email{chenxuwen@math.umd.edu}
\thanks{ The author thanks Matei Machedon for discussions related to this
work.}
\date{}
\maketitle

\begin{abstract}
This note shows that, with a little modification, the results in \cite{GMM}
hold in the 3-body interactions case.
\end{abstract}

\section{Introduction}

\label{sec:intro}

Recently, a second-order correction to the usual tensor product (mean-field)
approximation for the Hamiltonian evolution of a many-particle system with
2-body potential was found in Grillakis-Machedon-Margetis\cite{GMM}. Here is
their main result,

\begin{theorem}
\cite{GMM}\label{GMM Result} Let $\phi $ be a smooth solution of the Hartree
equation 
\begin{equation}
i\frac{\partial \phi }{\partial t}+\Delta \phi +(v^{\prime }\ast |\phi
|^{2})\phi =0  \label{3-Hartree}
\end{equation}%
with initial conditions $\phi _{0}$ and potential $v^{\prime }$ being even.
Assuming the following:

\begin{enumerate}
\item Let $k(t,x,y)\in L_{s}^{2}(dxdy)$ for a.e. $t$, solve the
Grillakis-Machedon-Margetis equation 
\begin{equation}
(iu_{t}+ug^{T}+gu-(I+p)m)=(ip_{t}+[g,p]+u\overline{m})(1+p)^{-1}u\text{ },
\label{GMM EQ}
\end{equation}%
where products in equation \ref{GMM EQ} mean compositions of operators, and 
\begin{align*}
& u(t,x,y):=\sinh (k):=k+\frac{1}{3!}k\overline{k}k+\ldots ~, \\
& \delta (x-y)+p(t,x,y):=\cosh (k):=\delta (x-y)+\frac{1}{2!}k\overline{k}%
+\ldots ~, \\
& g(t,x,y):=-\Delta _{x}\delta (x-y)-v(x-y)\phi (t,x)\overline{\phi }%
(t,y)-(v\ast |\phi |^{2})(t,x)\delta (x-y)~, \\
& m(t,x,y):=v(x-y)\overline{\phi }(t,x)\overline{\phi }(t,y)~.
\end{align*}

\item Functions $f(t):=\Vert e^{B}[A,V^{\prime }]e^{-B}\Omega \Vert _{{%
\mathcal{F}}}$ and $g(t):=\Vert e^{B}V^{\prime }e^{-B}\Omega \Vert _{{%
\mathcal{F}}}$ are locally integrable, where 
\begin{equation}
B(t)=\frac{1}{2}\int \left( k(t,x,y)a_{x}a_{y}-\overline{k}%
(t,x,y)a_{x}^{\ast }a_{y}^{\ast }\right) dxdy  \label{B-Def}
\end{equation}%
and ${\mathcal{F}}$ the Fock space defined in \cite{GMM}.

\item $\int d(t,x,x)\ dx$ is locally integrable in time, where 
\begin{align}
d(t,x,y)=& \left( i\sinh (k)_{t}+\sinh (k)g^{T}+g\sinh (k)\right) \overline{%
\sinh (k)}  \label{D} \\
-& \left( i\cosh (k)_{t}+[g,\cosh (k)]\right) \cosh (k)  \notag \\
-& \sinh (k)\overline{m}\cosh (k)-\cosh (k)m\overline{\sinh (k)}~.  \notag
\end{align}%
Then, there exist real functions $\chi _{0}$, $\chi _{1}$ such that%
\begin{align}
& \Vert e^{-\sqrt{N}A(t)}e^{-B(t)}e^{-i\int_{0}^{t}(N\chi _{0}(s)+\chi
_{1}(s))ds}\Omega -e^{itH_{N}^{\prime }}e^{-\sqrt{N}A(0)}\Omega \Vert _{{%
\mathcal{F}}}  \label{GMM estimate} \\
& \leq \frac{\int_{0}^{t}f(s)ds}{\sqrt{N}}+\frac{\int_{0}^{t}g(s)ds}{N}\text{
},  \notag
\end{align}

where $\Omega =(1,0,0,\cdots )\in {\mathcal{F}}$ is the vacuum state, 
\begin{equation}
A(\phi )=a(\overline{\phi })-a^{\ast }(\phi )  \label{A}
\end{equation}%
with $a$ and $a^{\ast }$ are the annihilation operator and creation
operator, and the Hamiltonian 
\begin{eqnarray*}
H_{N}^{\prime } &=&\int a_{x}^{\ast }\Delta a_{x}dx+\frac{1}{2N}\int
v^{\prime }(x-y)a_{x}^{\ast }a_{y}^{\ast }a_{x}a_{y}\ dxdy \\
&=&:H_{0}+\frac{1}{2N}V^{\prime }\text{ }.
\end{eqnarray*}
\end{enumerate}
\end{theorem}

\bigskip To prove Theorem \ref{GMM Result}, the key is to notice that, for $%
\phi $ and $k$ satisfying the assumptions, there is some "nice" Hermitian $%
\widetilde{L}$ such that 
\begin{equation*}
\frac{1}{i}\frac{\partial }{\partial t}\Psi =(\widetilde{L}-N\chi _{0}-\chi
_{1})\Psi \text{ },
\end{equation*}%
and%
\begin{equation*}
\Vert \widetilde{L}\Omega \Vert _{{\mathcal{F}}}\leq N^{-1/2}\Vert
e^{B}[A,V]e^{-B}\Omega \Vert _{{\mathcal{F}}}+N^{-1}\Vert e^{B}Ve^{-B}\Omega
\Vert _{{\mathcal{F}}}\text{ ,}
\end{equation*}%
where 
\begin{equation*}
\Psi (t)=e^{B(t)}e^{\sqrt{N}A(t)}e^{itH}e^{-\sqrt{N}A(0)}\Omega \text{ }.
\end{equation*}
Then applying the energy estimate to 
\begin{equation*}
\left( \frac{1}{i}\frac{\partial }{\partial t}-\widetilde{L}\right)
(e^{i\int_{0}^{t}(N\chi _{0}(s)+\chi _{1}(s))ds}\Psi -\Omega )=\widetilde{L}%
\Omega ~,
\end{equation*}%
would give the estimate \ref{GMM estimate}.

With the same set up, this result can be extended to the case of 3-body
interaction potentials.

\section{Main Result}

\label{sec:Main Result}

We consider the Fock space Hamiltonian $H_{N}:{\mathcal{F}}\rightarrow {%
\mathcal{F}}$ with the 3-body interaction potential defined by 
\begin{align}
H_{N}& =\int a_{x}^{\ast }\Delta a_{x}dx+\frac{1}{6N^{2}}\int
v(x-y,y-z)a_{x}^{\ast }a_{y}^{\ast }a_{z}^{\ast }a_{x}a_{y}a_{z}\ dx\,dy
\label{3H} \\
& =:H_{0}+\frac{1}{6N^{2}}V~,  \notag
\end{align}%
in correspondence to the Hamiltonian 
\begin{equation*}
H_{N,N}=\sum_{j=1}^{n}\Delta _{x_{j}}+\frac{1}{N^{2}}%
\sum_{i<j<k}^{n}v(x_{i}-x_{j},x_{j}-x_{k})~.
\end{equation*}%
We assume that $v$ is symmetric in $x$, $y$, and $z$. As shown in (1.2) and
(1.3) in Chen-Pavlovi\'{c}\cite{TChenAndNP}, every translation invariant
3-body potential can be written in the form $v(x-y,y-z)$. However, I don't
assume $v\geqslant 0$ as in \cite{TChenAndNP}.

\begin{theorem}
\label{main} With the same assumptions as in theorem \ref{GMM Result} with $%
H_{N}^{\prime }$ replaced by $H_{N}$, $V^{\prime }$ by $V$ in formula \ref%
{3H}. We change equation \ref{3-Hartree} to 
\begin{equation}
i\frac{\partial }{\partial t}\phi +\triangle \phi +\frac{1}{2}\phi \int
v(x-y,y-z)\left\vert \phi (y)\right\vert ^{2}\left\vert \phi (z)\right\vert
^{2}dydz=0  \label{5-Hartree}
\end{equation}

and 
\begin{align}
& g(t,x,y):=-\triangle \delta _{12}-\bigg(\int v_{12,23}\left\vert \phi
_{3}\right\vert ^{2}dz\bigg)\overline{\phi }_{1}\phi _{2}-\frac{1}{2}\bigg(%
\int v_{12,23}\left\vert \phi _{2}\right\vert ^{2}\left\vert \phi
_{3}\right\vert ^{2}dydz\bigg)\delta _{12}  \notag \\
& m(t,x,y):=\bigg(\int v_{12,23}\left\vert \phi _{3}\right\vert ^{2}dz\bigg)%
\overline{\phi }_{1}\overline{\phi }_{2}  \notag
\end{align}%
where $v_{12,23}\left\vert \phi _{3}\right\vert ^{2}$ is an abbreviation for
the product $v(x-y,y-z)\left\vert \phi (z)\right\vert ^{2}$,

Further, we assume that the functions 
\begin{equation*}
h(t):=\Vert e^{B}[A,[A,V]]e^{-B}\Omega \Vert _{{\mathcal{F}}}
\end{equation*}%
and%
\begin{equation*}
i(t):=\Vert e^{B}[A,[A,[A,V]]]e^{-B}\Omega \Vert _{{\mathcal{F}}}
\end{equation*}%
are also locally integrable.

Then, there exist real functions $\chi _{0}$, $\chi _{1}$ such that 
\begin{align*}
& \Vert e^{-\sqrt{N}A(t)}e^{-B(t)}e^{-i\int_{0}^{t}(N\chi _{0}(s)+\chi
_{1}(s))ds}\Omega -e^{itH_{N}}e^{-\sqrt{N}A(0)}\Omega \,\Vert _{{\mathcal{F}}%
} \\
& \leq \frac{\int_{0}^{t}f(s)ds}{6N^{\frac{3}{2}}}+\frac{\int_{0}^{t}g(s)ds}{%
6N^{2}}+\frac{\int_{0}^{t}h(s)ds}{12N}+\frac{\int_{0}^{t}i(s)ds}{36N^{\frac{1%
}{2}}}~.
\end{align*}
\end{theorem}

\begin{proof}
Similar to the proof of Theorem \ref{GMM Result}.
\end{proof}

\begin{remark}
Recall that, the equation for $k$ is \ref{GMM EQ}%
\begin{equation}
(iu_{t}+ug^{T}+gu-(I+p)m)=(ip_{t}+[g,p]+u\overline{m})(1+p)^{-1}u\text{ }.
\label{k-equation}
\end{equation}

A derivation of this equation will be given in Section \ref{sec:k-eq}.
\end{remark}

\begin{remark}
In Section \ref{sec:k-eq}, we will show that the $\widetilde{L}$ in this
case is 
\begin{align}
\widetilde{L}=& H_{0}+\int v_{12,23}\left\vert \phi _{3}\right\vert ^{2}%
\overline{\phi }_{2}\phi _{1}a_{x}^{\ast }a_{y}dxdydz+  \label{New L tutar}
\\
& \frac{1}{2}\int v_{12,23}\left\vert \phi _{2}\right\vert ^{2}\left\vert
\phi _{3}\right\vert ^{2}a_{x}^{\ast }a_{x}dxdydz  \notag \\
-& \int d(t,x,y)a_{y}^{\ast }a_{x}dx+\frac{1}{6N^{2}}e^{B}Ve^{-B}+  \notag \\
& \frac{1}{6N^{3/2}}e^{B}[A,V]e^{-B}+\frac{1}{12N}e^{B}[A,[A,V]]e^{-B} 
\notag \\
& +\frac{1}{36N^{\frac{1}{2}}}e^{B}[A,[A,[A,V]]]e^{-B}\text{ },  \notag
\end{align}

where $d(t,x,y)$ is the same as equation \ref{D}.
\end{remark}

\section{The New Hartree equation}

\label{sec:hartree}

In this section we derive the Hartree equation \ref{5-Hartree} for the
one-particle wave function $\phi $ as needed in Theorem \ref{sec:Main Result}%
.

\begin{lemma}
\label{Hartree-F} The following relations hold, where $A$ denotes $A(\phi )$%
, and $A$ , $V$ are defined by formulas \ref{A} and \ref{3H}): 
\begin{eqnarray*}
&&[A,V] \\
&=&3\int v(x-y,y-z)(\overline{\phi }(x)a_{y}^{\ast }a_{z}^{\ast
}a_{x}a_{y}a_{z}+\phi (x)a_{x}^{\ast }a_{y}^{\ast }a_{z}^{\ast
}a_{y}a_{z})dxdydz
\end{eqnarray*}%
\begin{eqnarray*}
&&[A,[A,V]] \\
&=&6\int v(x-y,y-z)(\overline{\phi }(x)\overline{\phi }(y)a_{z}^{\ast
}a_{x}a_{y}a_{z}+2\phi (x)\overline{\phi }(y)a_{x}^{\ast }a_{z}^{\ast
}a_{y}a_{z} \\
&&+\phi (x)\phi (y)a_{x}^{\ast }a_{y}^{\ast }a_{z}^{\ast }a_{y}a_{z})dxdydz
\\
&&+6\int v(x-y,y-z)\left\vert \phi (x)\right\vert ^{2}a_{y}^{\ast
}a_{z}^{\ast }a_{y}a_{z}dxdydz
\end{eqnarray*}%
\begin{eqnarray*}
&&[A,[A,[A,V]]] \\
&=&36\int v(x-y,y-z)\left\vert \phi (x)\right\vert ^{2}(\overline{\phi }%
(y)a_{z}^{\ast }a_{y}a_{z}+\phi (y)a_{y}^{\ast }a_{z}^{\ast }a_{z})dxdydz \\
&&+6\int v(x-y,y-z)(\overline{\phi }(x)\overline{\phi }(y)\overline{\phi }%
(z)a_{x}a_{y}a_{z}+\phi (x)\phi (y)\phi (z)a_{x}^{\ast }a_{y}^{\ast
}a_{z}^{\ast })dxdydz \\
&&+18\int v(x-y,y-z)(\overline{\phi }(x)\overline{\phi }(y)\phi
(z)a_{z}^{\ast }a_{x}a_{y}+\phi (x)\phi (y)\overline{\phi }(z)a_{x}^{\ast
}a_{y}^{\ast }a_{z})dxdydz
\end{eqnarray*}%
\begin{eqnarray*}
&&[A,[A,[A,[A,V]]]] \\
&=&72\int v(x-y,y-z)\left\vert \phi (x)\right\vert ^{2}(\overline{\phi }(y)%
\overline{\phi }(z)a_{y}a_{z}+\phi (y)\phi (z)a_{y}^{\ast }a_{z}^{\ast
})dxdydz \\
&&+144\int v(x-y,y-z)\left\vert \phi (x)\right\vert ^{2}\overline{\phi }%
(y)\phi (z)a_{z}^{\ast }a_{y}dxdydz \\
&&+72\int v(x-y,y-z)\left\vert \phi (x)\right\vert ^{2}\left\vert \phi
(y)\right\vert ^{2}a_{z}^{\ast }a_{z}dxdydz
\end{eqnarray*}%
\begin{eqnarray*}
&&[A,[A,[A,[A,[A,V]]]]] \\
&=&360\int v(x-y,y-z)\left\vert \phi (x)\right\vert ^{2}\left\vert \phi
(y)\right\vert ^{2}(\overline{\phi }(z)a_{z}+\phi (z)a_{z}^{\ast })dxdydz
\end{eqnarray*}%
\begin{eqnarray*}
&&[A,[A,[A,[A,[A,[A,V]]]]]] \\
&=&720\int v(x-y,y-z)\left\vert \phi (x)\right\vert ^{2}\left\vert \phi
(y)\right\vert ^{2}\left\vert \phi (z)\right\vert ^{2}dxdydz
\end{eqnarray*}
\end{lemma}

\begin{proof}
This is a direct calculation.
\end{proof}

Now, we write $\Psi _{1}(t)=e^{\sqrt{N}A(t)}e^{itH}e^{-\sqrt{N}A(0)}\Omega $
for which we carry out the calculation in the spirit of equation~(3.7) in
Rodnianski-Schlein \cite{Rod-S}.

\begin{proposition}
Say $\phi $ satisfies the Hartree equation 
\begin{equation*}
i\frac{\partial }{\partial t}\phi +\triangle \phi +\frac{1}{2}\phi \int
v(x-y,y-z)\left\vert \phi (y)\right\vert ^{2}\left\vert \phi (z)\right\vert
^{2}dydz=0
\end{equation*}%
then $\Psi _{1}(t)$ satisfies 
\begin{align*}
& \frac{1}{i}\frac{\partial }{\partial t}\Psi _{1}(t)=\bigg(H_{N}+\frac{1}{4!%
}\frac{1}{6}[A,[A,[A,[A,V]]]]+\frac{1}{6}N^{-3/2}[A,V] \\
& +\frac{1}{12}N^{-1}[A,[A,V]]+\frac{1}{36}N^{-\frac{1}{2}}[A,[A,[A,V]]] \\
& -\frac{N}{3}\int v(x-y,y-z)\left\vert \phi (x)\right\vert ^{2}\left\vert
\phi (y)\right\vert ^{2}\left\vert \phi (z)\right\vert ^{2}dxdydz\bigg)\Psi
_{1}(t)~.
\end{align*}
\end{proposition}

\begin{proof}
Applying the formulas 
\begin{equation*}
\left( \frac{\partial }{\partial t}e^{C(t)}\right) \left( e^{-C(t)}\right) =%
\dot{C}+\frac{1}{2!}[C,\dot{C}]+\frac{1}{3!}\big[C,[C,\dot{C}]\big]+\ldots
\end{equation*}
\begin{equation*}
e^{C}He^{-C}=H+[C,H]+\frac{1}{2!}\big[C,[C,H]\big]+\ldots ~.
\end{equation*}
to $C=\sqrt{N}A$ we have

\begin{equation}
\frac{1}{i}\frac{\partial }{\partial t}\psi _{1}(t)=L_{1}\psi _{1}~,
\label{r-s}
\end{equation}%
where 
\begin{align*}
& L_{1}=\frac{1}{i}\left( \frac{\partial }{\partial t}e^{\sqrt{N}%
A(t)}\right) e^{-\sqrt{N}A(t)}+e^{\sqrt{N}A(t)}H_{N}e^{-\sqrt{N}A(t)} \\
=& \frac{1}{i}\left( N^{1/2}\dot{A}+\frac{N}{2}[A,\dot{A}]\right)
+H_{N}+N^{1/2}[A,H_{0}]+\frac{N}{2!}[A,[A,H_{0}]]+\frac{1}{6}\bigg(%
N^{-3/2}[A,V] \\
& +\frac{N^{-1}}{2!}[A,[A,V]]+\frac{N^{-\frac{1}{2}}}{3!}[A,[A,[A,V]]] \\
& +\frac{1}{4!}[A,[A,[A,[A,V]]]]+\frac{N^{\frac{1}{2}}}{5!}%
[A,[A,[A,[A,[A,V]]]]] \\
& +\frac{N}{6!}[A,[A,[A,[A,[A,[A,V]]]]]]\bigg)\text{ }.
\end{align*}

The Hartree equation\ref{5-Hartree} is equivalent to setting%
\begin{equation}
\frac{1}{i}\dot{A}+[A,H_{0}]+\frac{1}{6}\frac{1}{5!}[A,[A,[A,[A,[A,V]]]]]=0~.
\label{H-F}
\end{equation}%
Or more explicitly, equation \ref{H-F} is 
\begin{equation*}
a(\overline{i\phi _{t}}+\overline{\triangle \phi }+\frac{1}{2}\overline{\phi 
}\int v_{12,23}\left\vert \phi _{2}\right\vert ^{2}\left\vert \phi
_{3}\right\vert ^{2}dydz)+a^{\ast }(i\phi _{t}+\triangle \phi +\frac{1}{2}%
\phi \int v_{12,23}\left\vert \phi _{2}\right\vert ^{2}\left\vert \phi
_{3}\right\vert ^{2}dydz)=0~,
\end{equation*}%
because of Lemma \ref{Hartree-F} and the fact that $[\Delta
_{x}a_{x},a_{y}^{\ast }]=(\Delta \delta )(x-y)$.

Thus 
\begin{equation*}
\frac{1}{i}[A,\dot{A}]+[A,[A,H_{0}]]+\frac{1}{5!}\frac{1}{6}%
[A,[A,[A,[A,[A,[A,V]]]]]]=0~,
\end{equation*}%
i.e \ref{r-s} simplifies to 
\begin{align}
& \frac{1}{i}\frac{\partial }{\partial t}\Psi _{1}(t)=\bigg(H_{N}+\frac{1}{4!%
}\frac{1}{6}[A,[A,[A,[A,V]]]]+\frac{1}{6}N^{-3/2}[A,V]  \label{terms} \\
& +\frac{1}{12}N^{-1}[A,[A,V]]+\frac{1}{36}N^{-\frac{1}{2}}[A,[A,[A,V]]] 
\notag \\
& -\frac{N}{3}\int v(x-y,y-z)\left\vert \phi (x)\right\vert ^{2}\left\vert
\phi (y)\right\vert ^{2}\left\vert \phi (z)\right\vert ^{2}dxdydz\bigg)\Psi
_{1}(t)~.  \notag
\end{align}
\end{proof}

Write the last term as 
\begin{equation*}
-\frac{N}{3}\int v(x-y,y-z)\left\vert \phi (x)\right\vert ^{2}\left\vert
\phi (y)\right\vert ^{2}\left\vert \phi (z)\right\vert ^{2}dxdydz:=-N\chi
_{0}~,
\end{equation*}%
then the first two terms on the right-hand side of \ref{terms} are the main
ones we need to consider, since the next four terms are $O\left( \frac{1}{%
N^{2}}\right) $, $O\left( \frac{1}{N^{\frac{3}{2}}}\right) ,$ $O\left( \frac{%
1}{N}\right) ,$and $O\left( \frac{1}{N^{\frac{1}{2}}}\right) $.

In order to kill the terms involving "only creation operators" i.e $%
a_{x}^{\ast }a_{y}^{\ast }$ in $\frac{1}{4!}\frac{1}{6}[A,[A,[A,[A,V]]]]$,
we introduce $B$ (see \ref{B-Def}) and let 
\begin{equation}
\Psi =e_{1}^{B}\Psi _{1}  \notag
\end{equation}%
as in \cite{GMM}. So we have 
\begin{equation}
\frac{1}{i}\frac{\partial }{\partial t}\Psi =L\Psi ~,  \notag
\end{equation}%
where

\begin{align*}
& L=\frac{1}{i}\left( \frac{\partial }{\partial t}e^{B}\right)
e^{-B}+e^{B}L_{1}e^{-B} \\
& =L_{Q}+\frac{1}{6N^{2}}e^{B}Ve^{-B}+\frac{1}{6}N^{-3/2}e^{B}[A,V]e^{-B}+%
\frac{1}{12}N^{-1}e^{B}[A,[A,V]]e^{-B} \\
& +\frac{1}{36}N^{-\frac{1}{2}}e^{B}[A,[A,[A,V]]]e^{-B}-N\chi _{0}~,
\end{align*}%
and 
\begin{equation}
L_{Q}=\frac{1}{i}\left( \frac{\partial }{\partial t}e^{B}\right)
e^{-B}+e^{B}\left( H_{0}+\frac{1}{4!}\frac{1}{6}[A,[A,[A,[A,V]]]]\right)
e^{-B}\text{ }.  \label{LQ}
\end{equation}

\section{Equation for $k$}

\label{sec:k-eq}

In this section, we derive the equation for $k$.

Let ${\mathcal{I}}:sp\rightarrow Quad$ be the Lie algebra isomorphism
defined by $Q(d,k,l)={\mathcal{I}}(S(d,k,l))$ in theorem 3 in \cite{GMM},
where $sp$ is the infinite dimensional Lie algebra containing%
\begin{equation*}
S(d,k,l)=\left( 
\begin{matrix}
d & k \\ 
l & -d^{T}%
\end{matrix}%
\right)
\end{equation*}%
and $k$ and $l$ are symmetric, and $Quad$ is the Lie algebra consisting of
quadratics of the form 
\begin{align*}
Q(d,k,l):=& \frac{1}{2}\left( 
\begin{matrix}
a_{x} & a_{x}^{\ast }%
\end{matrix}%
\right) \left( 
\begin{matrix}
d & k \\ 
l & -d^{T}%
\end{matrix}%
\right) \left( 
\begin{matrix}
-a_{y}^{\ast } \\ 
a_{y}%
\end{matrix}%
\right) \\
& =-\int d(x,y)\frac{a_{x}a_{y}^{\ast }+a_{y}^{\ast }a_{x}}{2}\ dx\,dy+\frac{%
1}{2}\int k(x,y)a_{x}a_{y}\ dx\,dy \\
& -\frac{1}{2}\int l(x,y)a_{x}^{\ast }a_{y}^{\ast }\ dx\,dy\text{ }.
\end{align*}

Recall that 
\begin{eqnarray*}
&&\frac{1}{4!6}[A,[A,[A,[A,V]]]] \\
&=&\frac{1}{2}\int v(x-y,y-z)\left\vert \phi (x)\right\vert ^{2}(\overline{%
\phi }(y)\overline{\phi }(z)a_{y}a_{z}+\phi (y)\phi (z)a_{y}^{\ast
}a_{z}^{\ast })dxdydz \\
&&+\int v(x-y,y-z)\left\vert \phi (x)\right\vert ^{2}\overline{\phi }(y)\phi
(z)a_{z}^{\ast }a_{y}dxdydz \\
&&+\frac{1}{2}\int v(x-y,y-z)\left\vert \phi (x)\right\vert ^{2}\left\vert
\phi (y)\right\vert ^{2}a_{z}^{\ast }a_{z}dxdydz\text{ ,}
\end{eqnarray*}%
we write%
\begin{equation*}
G=\left( 
\begin{matrix}
g & 0 \\ 
0 & -g^{T}%
\end{matrix}%
\right) \qquad \mbox{
and}\qquad M=\left( 
\begin{matrix}
0 & m \\ 
-\overline{m} & 0%
\end{matrix}%
\right)
\end{equation*}%
with 
\begin{equation*}
g=-\triangle \delta _{12}-\bigg(\int v_{12,23}\left\vert \phi
_{3}\right\vert ^{2}dz\bigg)\overline{\phi }_{1}\phi _{2}-\frac{1}{2}\bigg(%
\int v_{12,23}\left\vert \phi _{2}\right\vert ^{2}\left\vert \phi
_{3}\right\vert ^{2}dydz\bigg)\delta _{12}
\end{equation*}%
and%
\begin{equation*}
m=\bigg(\int v_{12,23}\left\vert \phi _{3}\right\vert ^{2}dz\bigg)\overline{%
\phi }_{1}\overline{\phi }_{2}\text{ .}
\end{equation*}

Though it's true that 
\begin{equation*}
B={\mathcal{I}}(K)\text{ },
\end{equation*}%
for 
\begin{equation}
K=\left( 
\begin{matrix}
0 & k(t,x,y) \\ 
\overline{k}(t,x,y) & 0%
\end{matrix}%
\right) \text{ },  \label{K-Def}
\end{equation}%
it's not quite true that

\begin{equation*}
H_{0}+\frac{1}{4!}\frac{1}{6}[A,[A,[A,[A,V]]]]={\mathcal{I}}(G+M)\text{ .}
\end{equation*}%
However, the commutators of ${\mathcal{I}}(G+M)$ and $H_{0}+\frac{1}{4!}%
\frac{1}{6}[A,[A,[A,[A,V]]]]$ with $B$ are the same as in the discussion in
page 15 in \cite{GMM}.

Now, $L_{Q}$ reads 
\begin{eqnarray*}
L_{Q} &=&\frac{1}{i}\left( \frac{\partial }{\partial t}e^{B}\right)
e^{-B}+e^{B}\left( H_{0}+\frac{1}{4!}\frac{1}{6}[A,[A,[A,[A,V]]]]\right)
e^{-B} \\
&=&H_{G}+{\mathcal{I}}\left( \left( \frac{1}{i}\frac{\partial }{\partial t}%
e^{K}\right) e^{-K}+[e^{K},G]e^{-K}+e^{K}Me^{-K}\right) \\
&=&H_{G}+{\mathcal{I}}(\mathcal{M}_{1}+\mathcal{M}_{2}+\mathcal{M}_{3})~,
\end{eqnarray*}%
where 
\begin{equation*}
H_{G}=H_{0}+\int v_{12,23}\left\vert \phi _{3}\right\vert ^{2}\overline{\phi 
}_{2}\phi _{1}a_{x}^{\ast }a_{y}dxdydz+\frac{1}{2}\int v_{12,23}\left\vert
\phi _{2}\right\vert ^{2}\left\vert \phi _{3}\right\vert ^{2}a_{x}^{\ast
}a_{x}dxdydz,
\end{equation*}%
in this setting.

Then by the definition of the isomorphism ${\mathcal{I}}$, the coefficient
of $a_{x}a_{y}$ is $-(\mathcal{M}_{1,12}+\mathcal{M}_{2,12}+\mathcal{M}%
_{3,12}),$ and the coefficient of $a_{x}^{\ast }a_{y}^{\ast }$ is $(\mathcal{%
M}_{1,21}+\mathcal{M}_{2,21}+\mathcal{M}_{3,21}).$ To write it explicitly:%
\begin{align}
& -(\mathcal{M}_{1,12}+\mathcal{M}_{2,12}+\mathcal{M}_{3,12})  \label{coeff}
\\
& =\overline{(\mathcal{M}_{1,21}+\mathcal{M}_{2,21}+\mathcal{M}_{3,21})} 
\notag \\
=& (i\sinh (k)_{t}+\sinh (k)g^{T}+g\sinh (k))\overline{\cosh (k)}-  \notag \\
& (i\cosh (k)_{t}-[\cosh (k),g])\sinh (k)  \notag \\
& -\sinh (k)\overline{m}\sinh (k)-\cosh (k)m\overline{\cosh (k)}~.  \notag
\end{align}%
Setting formula \ref{coeff} to $0$ gives the equation \ref{k-equation}.
Hence we have the following theorem similar to Corollary 1 in \cite{GMM}.

\begin{theorem}
\label{maincor} If $\phi $ and $k$ solve \ref{5-Hartree} and \ref{k-equation}%
, then the coefficients of $a_{x}a_{y}$ and $a_{x}^{\ast }a_{y}^{\ast }$
drop out and $L_{Q}$ becomes 
\begin{align*}
L_{Q}=& H_{0}+\int v_{12,23}\left\vert \phi _{3}\right\vert ^{2}\overline{%
\phi }_{2}\phi _{1}a_{x}^{\ast }a_{y}dxdydz+\frac{1}{2}\int
v_{12,23}\left\vert \phi _{2}\right\vert ^{2}\left\vert \phi _{3}\right\vert
^{2}a_{x}^{\ast }a_{x}dxdydz \\
-& \int d(t,x,y)\frac{a_{x}a_{y}^{\ast }+a_{y}^{\ast }a_{x}}{2}\ dxdy~,
\end{align*}%
where $d$ is given by formula \ref{D} and the full operator reads 
\begin{align*}
L=& H_{0}+\int v_{12,23}\left\vert \phi _{3}\right\vert ^{2}\overline{\phi }%
_{2}\phi _{1}a_{x}^{\ast }a_{y}dxdydz+\frac{1}{2}\int v_{12,23}\left\vert
\phi _{2}\right\vert ^{2}\left\vert \phi _{3}\right\vert ^{2}a_{x}^{\ast
}a_{x}dxdydz- \\
& \int d(t,x,y)a_{y}^{\ast }a_{x}dx+\frac{1}{6N^{2}}e^{B}Ve^{-B}+\frac{1}{%
6N^{3/2}}e^{B}[A,V]e^{-B}+\frac{1}{12N}e^{B}[A,[A,V]]e^{-B} \\
& +\frac{1}{36N^{\frac{1}{2}}}e^{B}[A,[A,[A,V]]]e^{-B}-N\chi _{0}-\chi _{1}
\\
& :=\widetilde{L}-N\chi _{0}-\chi _{1}~,
\end{align*}%
and 
\begin{equation*}
\chi _{0}=\frac{1}{3}\int v(x-y,y-z)\left\vert \phi (x)\right\vert
^{2}\left\vert \phi (y)\right\vert ^{2}\left\vert \phi (z)\right\vert
^{2}dxdydz,
\end{equation*}%
\begin{equation*}
\chi _{1}(t)=-\frac{1}{2}\int d(t,x,x)dx~.
\end{equation*}%
This matches equation \ref{New L tutar}.
\end{theorem}

Notice that 
\begin{eqnarray*}
\widetilde{L}\Omega &=&\bigg(\frac{1}{6N^{2}}e^{B}Ve^{-B}+\frac{1}{6}%
N^{-3/2}e^{B}[A,V]e^{-B}+ \\
&&\frac{1}{12}N^{-1}e^{B}[A,[A,V]]e^{-B}+\frac{1}{36}N^{-\frac{1}{2}%
}e^{B}[A,[A,[A,V]]]e^{-B}\bigg)\Omega ~,
\end{eqnarray*}%
in our case. This ends the proof of theorem \ref{main}.

The verification that the hypotheses of the main theorem are satisfied by
some singular potentials will be provided in a subsequent article.

\end{document}